\documentclass[11pt]{article}
\usepackage{color}
\usepackage{times}
\usepackage{amssymb}
\usepackage{amsmath,amsthm}
\usepackage{graphicx}
\usepackage{multirow}
\usepackage[sans]{dsfont}
\usepackage[authoryear]{natbib}
\usepackage{tabularx}
\usepackage[text={6in,9in},centering]{geometry}

\newcommand{\bbE}{\mathbb{E}}
\DeclareMathOperator{\diag}{\text{diag}}
\DeclareMathOperator{\im}{\text{im}}
\DeclareMathOperator{\var}{\text{\sf var}}
\newcommand{\bbR}{\mathbb{R}}

\newcommand{\bx}{\mathbf{x}}
\newcommand{\by}{\mathbf{y}}
\newcommand{\bw}{\mathbf{w}}
\newcommand{\bB}{\mathbf{B}}

\newcommand{\Xt}{\widetilde{X}_t}
\newcommand{\Xb}{\overline{X}_t}

\newcommand{\sigmat}{\tilde{\sigma}}

\newcommand{\scal}[2]{\langle{#1},{#2}\rangle}
\newcommand{\nor}[1]{\|{#1}\|}

\newcommand{\tr}{\top\!}
\newcommand{\bbone}{\mathds{1}}

\newtheorem{theorem}{Theorem}[section]
\newtheorem{corollary}{Corollary}[section]
\newtheorem{definition}{Definition}[section]

\newtheorem{lemma}{Lemma}[section]

\theoremstyle{definition}

\begin{document}
\vspace{10mm} \ \\ \ \\ \ \\

\begin{center}
{\bf \LARGE Synchronization and Redundancy:\\ Implications for Robustness of Neural Learning and Decision Making}
\end{center}

\ \\
{\bf \large Jake Bouvrie} ({\tt jvb@math.duke.edu}) \\
{Department of Mathematics, Duke University, Durham, NC USA} \\
\ \\
{\bf \large Jean-Jacques Slotine} ({\tt jjs@mit.edu})\\
{Nonlinear Systems Laboratory, Massachusetts Institute of Technology, Cambridge, MA USA}\\

\ \\ [-2mm]

\begin{abstract}
Learning and decision making in the brain are key processes critical to survival, and yet
are processes implemented by non-ideal biological building blocks which can impose significant
error. We explore quantitatively how the brain might cope with this inherent source of error by taking advantage of two ubiquitous mechanisms, redundancy and synchronization. In particular we consider a neural process whose goal is to learn a decision function by implementing a nonlinear gradient dynamics. The dynamics, however, are assumed to be corrupted by perturbations modeling the error which might be incurred due to limitations of the biology, intrinsic neuronal noise, and imperfect measurements.
We show that error, and the associated uncertainty surrounding a learned solution, can be controlled in large part by trading off synchronization strength among multiple redundant neural systems against the noise amplitude. The impact of the coupling between such redundant systems is quantified by the spectrum of the network Laplacian, and we discuss the role of network topology in synchronization and in reducing the effect of noise. A range of situations in which the mechanisms we model arise in brain science are discussed, and we draw attention to experimental evidence suggesting that cortical circuits capable of implementing the computations of interest here can be found on several scales. Finally, simulations comparing theoretical bounds to the relevant empirical quantities show that the theoretical estimates we derive can be tight.
\end{abstract}

\section{Introduction}
Learning and decision making in the brain are key processes critical to survival, and yet
are processes implemented by imperfect biological building blocks which can impose significant
error. We suggest that the brain can cope with this inherent source of error by taking advantage of two ubiquitous mechanisms: {\em redundancy}, and {\em sharing of information}. These concepts will be made precise in the context of a specific model and learning scenario which together can serve
as a conceptual tool for illustrating the effect of redundancy and sharing.

Motivated by the problem of learning to discriminate, we consider a neural process whose goal is to learn a decision function by implementing a nonlinear gradient dynamics. The dynamics, however, are assumed to be corrupted by perturbations modeling the error which might be incurred. This general perspective is
intended to capture a range of possible learning instances occurring at different anatomical scales: The neural process can involve whole brain areas communicating via behavioral, motor or sensory pathways~\citep{Schnitzler:NatRev:05}, as in the case of the multiple amygdala-thalamus loops assumed to underpin fear conditioning for instance~\citep{LeDoux:AnnRevNeuro:00,Maren:AnnRevNeuro:01}. Interacting local field potentials (LFPs) may also be modeled as both direct (by long range phase-locking, e.g. in  olfactory systems~\citep{Friedrich:NatNeuro:04}) or indirect measurements of coordination and interaction among large assemblies of neurons. The learning dynamics may alternatively model smaller ensembles of individual neurons, as in primary motor cortex, though we do not emphasize biological realism in our models at this scale. Nevertheless, one may still draw useful conclusions as to the role of redundancy and information sharing. The error too may be treated at different scales, and may take the form of noise intrinsic to the neural environment~\citep{Faisal:NRN:2008} on a large, aggregate scale (e.g. in the case of LFPs) or on a small scale involving localized populations of neurons.

If there is noise corrupting the learning process, an immediate question is whether it is possible to gauge the accuracy of the predictions of the learned function, and to what extent the organism can reduce uncertainty in its decisions by taking advantage of a simple, common information sharing mechanism.
If there is redundancy in the form of multiple independent {\em copies} of the dynamical circuit~\citep{Adams:98,Fernando:2010}, it is reasonable to expect that averaging over the different solutions might reduce noise via cancelation effects. In the case of learning in the brain, however, this approach is problematic because neurons are susceptible to saturation of their firing rates, and on large scales aggregate signal amplitudes will also saturate; the macroscopic dynamics that neuron populations and assemblies obey can be strongly nonlinear. When the dynamics followed by different dynamical systems are nonlinear, one cannot expect to gain a meaningful signal by linear averaging (see e.g.~\citep{Tabareau10} and examples therein). As a simple illustration of this phenomenon, consider a collection of noisy sinusoidal oscillators allowed to run starting from different initial conditions, with identical frequencies and independent noise terms.  The oscillators will be out of phase from each other, so an average over the trajectories will not yield anything close to a clean version of a sinusoid at the desired frequency. On the other hand it is reasonable to suppose that {\em synchronization} across neuron populations or between macro-scale cortical loops may provide sufficient phase alignment to make linear averaging, and thus ``consensus'', a powerful proposal for reducing the effects of noise~\citep{Tabareau10,Masuda10,Cao10,Young:ACC:10,Poul:ACC:10,Gigante:PLOS:09}. Indeed, it is a well known fact, that synchrony within a system of coupled dynamical elements provides (quantifiable) robustness to perturbations occurring in any one element's dynamics~\citep{Needleman:PD:2001,WangSlotine05,Pham:NN:07}.

We will place much emphasis on exploring quantitatively the role of synchrony in controlling uncertainty arising from noise modeling neural error. In particular, we base our work on the argument that noisy, nonlinear trajectories can be linearly averaged if fluctuations due to noise can be made small and that fluctuations can be made small by coupling the dynamical elements appropriately.  In the stochastic setting adopted here, ``synchronization'' refers to state synchrony: the tendency for individual elements' trajectories to move towards a common trajectory, in a quantifiable sense. The estimates we present directly characterize the tradeoff between the network's tendency towards synchrony and the noise, and ultimately address the specific role this tradeoff plays in determining uncertainty surrounding a function learned by an imperfect learning system.

We further show how and where the topology of the network of neural ensembles impacts the extent to which the noise, and therefore uncertainty, can be controlled. The estimates we provide take into account in a fundamental way both the nonlinearity in the dynamics and the noise. More generally, the work discussed here also has implications in other related domains, such as networks of coupled learners or adaptive sensor networks, and can be extended to multitask online or dynamic learning settings. The difficulty inherent in analyzing {\em dynamic} learning systems, such as hierarchies with feedback~\citep{Mumford:BC:92,LeeMumford03}, poses a challenge. But considering dynamic systems can yield substantial benefits: transients can be important, as suggested by the literature on regularization paths and early stopping~\citep{Yao07}. Furthermore, the role of feedback/backprojections and attention-like mechanisms in learning and recognition systems, both biological and artificial, is known to be important but is not well understood~\citep{Hahnloser:NN:99,Itti:NRN:01,Hung:Science:05}.

The paper is organized as follows. In Section~\ref{sec:setup} we consider a specific learning problem and define a system of stochastic differential equations (SDEs) modeling a simple dynamic learning process. We then discuss stability and network topology in the context of synchronization. In Section~\ref{sec:theorems} we present the main theoretical results of the paper, a set of uncertainty estimates, postponing proofs until later. Then in Section~\ref{sec:expts} we provide simulations and compare empirical estimates to the theoretical quantities predicted by the Theorems in Section~\ref{sec:theorems}. Section~\ref{sec:discussion} provides a discussion addressing the significance and applicability of our theoretical contributions to neuroscience and behavior. Finally, in Section~\ref{sec:proofs} we give proofs of the results stated in Section~\ref{sec:theorems}.

\section{Biological Learning as a Stochastic Network Model}\label{sec:setup}
The learning process we will model is that of a one-dimensional linear fitting problem described by gradient based minimization of a square loss objective, in the spirit of Rao \& Ballard~\citep{Rao:NNS:99}. This is perhaps the simplest and most fundamental abstract learning problem that an organism might be confronted with -- that of using experiential evidence to infer correlations and ultimately discover causal relationships which govern the environment and which can be used to make predictions about the future. The model realizing this learning process is also simple, in that we capture neural communication as an abstract process ``in which a neural element (a single neuron or a population of neurons) conveys certain aspects of its functional state to another neural element''~\citep{Schnitzler:NatRev:05}. In doing so, we focus on the underlying computations taking place in the nervous system, rather than dwell on neural representations. Even this simple setting becomes involved technically, and is rich enough to explore all of the key themes discussed above. Our model also supports nonlinear decision functions in the sense that we might consider taking a linear function of nonlinear variables whose values might be computed upstream. In this case the development would be similar, but extended to the multidimensional case. The model may also be extended to richer function classes and more exotic loss functions directly, however for our purposes the additional generality does not yield significant further insight and furthermore might raise biological plausibility concerns\footnote{In the sense that one would have to carefully justify biologically the particular nonlinearities going into a nonlinear decision function on a case-by-case basis.}.

\subsection{Problem Setup}
To make the setting more concrete, we begin by assuming that we have observed a set of input-output examples $\{x_i\in\bbR,y_i\in\bbR \}_{i=1}^m$, each
representing a generic unit of sensory experience, and want to estimate the linear regression function
$f_w(x) = wx$. Adopting the square loss, the total error incurred on the observations by $f_w$ is given by the familiar expression
\begin{equation*}\label{eqn:obective_fn}
E(w) = \sum_{i=1}^m(y_i - f_{w}(x_i))^2 = \sum_{i=1}^m(y_i - wx_i)^2.
\end{equation*}
We will model adaptation (training) by a noisy gradient descent process, with biologically plausible dynamics, on this squared prediction error loss function. The trajectory of the slope parameter over time $w(t)$ and its governing dynamics may be represented in the biology in various forms. Stochastic rate codes, average activities in populations of neurons and population codes, localized direct electrical signals and chemical concentration gradients are some possibilities occurring across a range of scales. The dynamical system may also be interpreted as modeling the noisy, time-varying strength of a local field potential or other macro electrophysiological signal when there are multiple, interacting brain regions. We discuss these possibilities further in Section~\ref{sec:discussion}.

The gradient of $E$ with respect to the weight parameter is given by $\nabla_{w} E =  -\sum_{i=1}^m(y_i - wx_i)x_i$, and serves as the starting point. The gradient dynamics $\dot{w}=-\nabla_{w} E(w)$ are both linear and noise-free. Following the discussion above, we modify these dynamics to capture nonlinear {\em saturation effects} as well as (often substantial) {\em noise} modeling error.
Saturation effects lead to a saturated gradient which we model in the form of the hyperbolic tangent nonlinearity,
$$
\dot{w} = -\tanh(a\nabla_{w} E(w)),
$$
where $a$ is a slope parameter. Note that the saturated dynamics need not be interpretable as itself
the gradient of an appropriate loss function. The fundamental learning problem is defined by the
square-loss, but it is implemented using an imperfect mechanism which imposes the nonlinearity\footnote{ Put differently, in our setting the nonlinearity is not part of the learning {\em problem}, and so the saturated gradient dynamics should not be viewed as the gradient of another error criteria}. The error is modeled with an additional diffusion (noise) term giving the SDE
\begin{equation}\label{eqn:noisy_single}
dw_t = -\tanh(a\nabla_{w} E(w_t))dt + \sigma dB_t ,
\end{equation}
where $dB_t$ denotes the standard 1-dimensional Wiener increment process with standard deviation  $\sigma > 0$.
As mentioned before, this noise term $\sigma dB_t$ and corresponding error is due to intrinsic neuronal noise~\citep{Faisal:NRN:2008} (aggregated or localized) and possible interference between large assemblies of neurons or circuits and parallels the more general concept of measurement error in networks of coupled dynamical systems.

\subsection{Synchronization and Noise}
We now consider the effect of having $n$ independent copies of the neural system or pathway implementing the dynamics~\eqref{eqn:noisy_single}, with associated parameters $\{w_1(t),\ldots,w_n(t)\}$. Since these dynamics are nonlinear, the effect of the noise cannot be reduced by simply averaging over the independent trajectories. However, if the circuit copies are coupled strongly enough they will attempt to {\em synchronize}, and averaging over the copies becomes a potentially powerful way to reduce the effect of the noise~\citep{Sherman:BJ:1991,Needleman:PD:2001}. The noise can be potentially large (we do not make any small-noise assumptions), and will of course act to break the synchrony. We will explore how well the noise can be reduced by synchronization and redundancy in the sections that follow.

Given $n$ {\em diffusively coupled} copies of the noisy neural system, and setting $a=1$ in~\eqref{eqn:noisy_single}, we have the following
system of nonlinear SDEs:
\begin{equation}\label{eqn:component_sys}
dw_i(t) = -\tanh\Biggl[\sum_{\ell=1}^m\bigl(w_i(t)x_{\ell} - y_{\ell}\bigr)x_{\ell}\Biggr]dt +
\sum_{j=1}^n W_{ij}(w_j-w_i)dt + \sigma dB_t^{(i)}
\end{equation}
for $i=1,\ldots,n$, where $B_t^{(i)}$ are {\em independent} standard Wiener processes.
The diffusive couplings here should be interpreted as modeling abstract intercommunication between and among different neural circuits, populations, or pathways. In such a general setting, diffusive coupling is a natural and mathematically tractable choice that can capture the key, aggregate aspects of communication among neural systems. Electrical connections such as those implemented by gap junctions in the mammalian cortex~\citep{Fukuda:06,Bennett:04} are also modeled well by diffusive coupling terms when individual neurons are being discussed, however we emphasize that the
system in Equation~\eqref{eqn:component_sys} is a conceptual model involving possibly large brain regions and do not make assumptions at a level of biological detail that would invoke or require gap-junction type connectivity.

Each copy of the basic neural circuit is corrupted by independent noise processes but follows the same noise-free dynamics as the others, modulo initial conditions. In fact these coupled systems may start from very different initial conditions. We will assume for simplicity uniform symmetric weights $W_{ji} = W_{ij} = \kappa > 0$ when element $i$ is connected to element $j$.
Defining $(\bw)_i$ to be the (scalar) output of the $i$-th circuit, we can rewrite the system~\eqref{eqn:component_sys} in vector form as
\begin{equation}\label{eqn:orig_sys}
d\bw(t) = -\Bigl(\tanh\Bigl[\sum_{i=1}^m\bigl(\bw(t)x_i - y_i\bbone\bigr)x_i\Bigr] +
L\bw(t)\Bigr)dt + \sigma d\bB_t
\end{equation}
where $L=\diag(W\bbone) - W$ is the {\em network Laplacian}, and $\mathbf{B}_t$ is the standard $n$-dimensional Wiener process. The spectrum of the network
Laplacian captures important properties of the network's topology, and will play a key role.
Finally, the change of variable $X_t:=\bw\nor{\bx}^2 - \scal{\bx}{\by}\bbone$, with $(\bx)_i = x_i, (\by)_i=y_i$, yields a system that will be easier to analyze:
\begin{equation}\label{eqn:simple_sys}
dX_t = -\bigl(\tanh(X_t)\nor{\bx}^2
+ LX_t\bigr)dt + \sigmat d\bB_t
\end{equation}
where we have defined $\sigmat:=\sigma\nor{\bx}^2$.
The unique globally stable equilibrium point for the deterministic part of~\eqref{eqn:simple_sys} is seen to be $X^{*}=0$, which checks with the fact that the solution to the linear regression problem is $w^*=\scal{\bx}{\by}/\scal{\bx}{\bx}$ in this simple case.

\subsection{Role of Network Topology}\label{sec:graphs_ex}
The topology of a network of dynamical systems strongly influences synchronization, to include
the rate at which elements synchronize and whether sync (or the tendency to sync) can occur at all in the first place. Thus the pattern of interconnections among neural systems plays an important role in controlling uncertainty by way of synchronization properties. In a network of stochastic systems of the general (diffusive) type described in Section~\ref{sec:setup}, topology can be seen to influence the robustness of synchrony to noise through the spectrum of the network Laplacian. Laplacians arising in various interesting networks and applications have received much attention, both in biological decision making and in the context of synchronization of dynamical systems more generally~\citep{Kopell:86,Kopell:00,Jadbabaie:03,WangSlotine05,Taylor:09,Poul:ACC:10}.

We will consider four important network graphs here, and these arrangements will be helpful examples to keep in mind when interpreting the results given in Section~\ref{sec:theorems}.
The simplest graph of coupled elements is perhaps the full, all-to-all graph. As one may guess, this network is also the easiest to synchronize since each element can speak directly to the others. The spectrum of the network Laplacian $\lambda(L)$ for this graph shows why it might be especially effective for reducing uncertainty in the context of Equation~\eqref{eqn:orig_sys}. With uniform coupling strength $\kappa>0$ and $n$ denoting the number of elements in the network, one can check that $\lambda(L) = \{0,n\kappa,\ldots,n\kappa\}$. Denote by $\lambda_{-}$ the smallest non-zero (Fiedler) eigenvalue, and by $\lambda_{+}$ the largest eigenvalue. Here $\lambda_{-}=\lambda_{+}=n\kappa$ and it is these eigenvalues that control synchronization for any given network. As we will show below in Theorem~\ref{thm:tilde_bounds}, the effect of the noise can be reduced particularly quickly precisely because the non-zero eigenvalues depend on both parameters, $\kappa$ and $n$.

\begin{figure}[t]
\centering
\includegraphics[height=2.5cm]{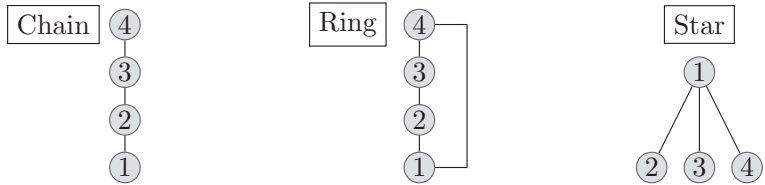}
\caption{{\em Examples of (undirected) network graphs.}}
\label{fig:ex_graphs}
\end{figure}

If fewer connections are made in the network it becomes harder to synchronize, and we move away from
the all-to-all ideal.  Figure~\ref{fig:ex_graphs} shows some other common network graphs. The undirected
ring graph, appearing in the middle, has spectrum $\lambda_i(L) =    2\kappa\bigl[1-\cos\bigl(\tfrac{2\pi}{n}(i-1)\bigr)\bigr], i=1,\ldots,n$. If the single edge connecting
the first and last elements is removed to make a chain as shown on the left in the Figure, the network
becomes considerably harder to synchronize~\citep{Kopell:86}, although the spectrum of the chain looks similar: $\lambda_i(L) =    2\kappa\bigl[1-\cos\bigl(\tfrac{\pi}{n}(i-1)\bigr)\bigr]$. This makes intuitive sense because information is constrained to flow through only one path, and with possibly significant delays. Finally, the star graph shown on the right in the Figure has spectrum $\lambda(L) = \{0,\kappa,\ldots,\kappa,n\kappa\}$, and we can see that the key Fiedler eigenvalue $\lambda_{-}=\kappa$ does not grow with the size of the network $n$. The Theorems in Section~\ref{sec:theorems} then predict that it will be impossible to increase the synchronization rate simply by incorporating more copies of the neural circuit. The coupling strength must also increase to make fluctuations from the common trajectory (synchronization subspace) small. We will discuss this case in more detail. As might be particularly relevant to brain anatomy, random graphs and directed graphs may also be considered, and have been
studied extensively~\citep{BollobasRandomGraphs}.

In neuroscience-related models, each connection in a network has an associated biophysical cost in terms of energy and space requirements. All-to-all networks, with $n^2$ connections among $n$ circuits or neurons, is often criticized as being biologically unrealistic because of this cost. However, it has been noted that all-to-all connectivity can be implemented with $2n$ connections using quorum sensing ideas~\citep{Taylor:09}, wherein a global average is computed and shared. The global average is computed given inputs from all $n$ elements, and this average is sent back to each circuit via another $n$ connections. The shared variable may be communicated by synapses, or sensed chemically or electrically. Although quorum sensing cannot realize {\em any} set of $n^2$ connections, the global average may be a weighted average or there may be several common variables organized hierarchically. This allows for a rich set of networks with $\mathcal{O}(n)$ connectivity which behave more like networks with all-to-all connectivity for synchronization and stability purposes. Furthermore, dynamics in the computation of the quorum variable itself, when appropriate for modeling purposes, does not necessarily pose any special difficulty for establishing synchronization properties if virtual systems are used~\citep{Russo:10}.

The difficulty with which synchrony may be imposed can be ``normalized'' by the number of connections in many cases to obtain a comparison between synchronization properties of various graphs that takes biological cost into account. Using quorum variables where appropriate, graphs whose spectrums depend on $n$ are thus roughly comparable on equal biological terms. Cost-normalized comparisons of synchronization properties are not always possible or meaningful, however. Consider the ring and chain networks introduced above. There is a difference of one edge between the two, but in the noise-free setting for example the chain requires asymptotically four times more effort to synchronize than the ring architecture (see e.g.~\citep{WangSlotine05}, Example 4.5).

\subsection{A Comment on Stability and Contraction}\label{sec:stability}
We turn to analyzing the stability of the nonlinear system given by Equation~\eqref{eqn:simple_sys}.
We will argue that this is difficult for two reasons: the presence of noise, and the fact that the (noise-free) dynamics saturate in magnitude. Indeed, without additional assumptions, one cannot in general show that the system is globally exponentially stable.
A common method for studying the stability properties of a noiseless nonlinear dynamical system
is via Lyapunov theory~\citep{SlotineBook}, however in the presence of noise system trajectories along the Lyapunov surface may not be strictly decreasing. Contraction analysis~\citep{Lohmiller98,WangSlotine05} is a differential formalism related to Lyapunov exponents, and captures the notion that a system is stable in some region if initial conditions or temporary disturbances are forgotten. If all neighboring trajectories converge to each other, global exponential convergence to a single trajectory can be concluded:
\begin{definition}[Contraction]
Given the system equations $\dot{\bx} = \mathbf{f}(\bx, t)$, a region of the state
space is called a contraction region if the Jacobian $J_f = \tfrac{\partial f}{\partial x}$ is
uniformly negative definite in that region. Furthermore, the contraction {\em rate} is given by $\beta$, where
$\tfrac{1}{2}(J_f + J_f^{\tr}) \leq \beta\mathbf{I} < 0$.
\end{definition}
An analogous definition in the case stochastic dynamics has also been developed~\citep{Pham09},
and requires contraction of the noise-free dynamics as well as a uniform upper bound on the
variance of the noise. However for the system~\eqref{eqn:simple_sys}, the Jacobian is found to be
$$
J(\bw) = \nor{\bx}^2\diag(\tanh^2(\bw)-1) - L
$$
so that
$
\lambda\bigl(J(\bw)\bigr) < -\lambda\bigl(L\bigr) \leq -\lambda_{\text{min}}(L) = 0 .
$
The subspace of constant vectors is a flow invariant subspace, and $L$ does not
contribute to the dynamics in this flow invariant space since $L$ has a zero eigenvalue
corresponding to its constant eigenvector. This difficulty can arise whenever one considers
 diffusively coupled elements, and in such cases the usual way
around this difficulty is to work with an auxiliary or virtual system (as in e.g.~\citep{Pham:NN:07}) and study contraction to the flow invariant subspace starting from initial conditions outside. However since $\tanh^{\prime}(x)=1-\tanh^2(x)$, we still are left with the difficulty that the noise-free dynamics can have a convergence rate to equilibrium arbitrarily close to zero as one travels far out to the tails of the $\tanh$ function; the system is not necessarily contracting. Indeed, for any saturated dynamics, $\tanh\bigl(\mathbf{f}(\bx,t)\bigr)$, the
 rate can be arbitrarily small. Thus one cannot easily determine the rate of convergence to
equilibrium using standard techniques. The analysis which we provide in the succeeding sections
will attempt to get around these difficulties by separately exploring the system's behavior in and out of the flow-invariant (synchronization) subspace of constant vectors.

\section{Controlling Uncertainty in Learning}\label{sec:theorems}
In this section we present and interpret the main results of the paper. The argument we
put forward is that noisy, nonlinear trajectories can be linearly averaged to reduce the noise if fluctuations due to noise can be made small. We show that the fluctuations can be made small by coupling the dynamical systems, and that one can precisely control the size of the fluctuations. In particular, we give estimates which show that the tradeoff between noise and coupling strength among neural circuits  determines the amount of uncertainty surrounding the decisions made by the neural system. Proofs of the Theorems are postponed until Section~\ref{sec:proofs}.

\subsection{Preliminaries}
We begin by decomposing the stochastic process $\{X_t\in\bbR^n\}_{t\geq 0}$ into a sum describing fluctuations about the center of mass. Let $P = I - (1/n)\bbone\bbone^{\tr}$, the canonical projection onto the zero-mean subspace of $\bbR^n$, and define $Q = I - P$. Then for all $t\geq 0$, $X_t = PX_t + QX_t$.
Clearly, $\ker P =\im Q$ is the subspace of constant vectors. We will adopt the notation
$\Xt$  for $PX_t$, and $\Xb\bbone$ for $QX_t$ (along with the analogous notation $\widetilde{\bw}_t$ and $\bar{w}_t$), and derive expressions for these quantities based on~Equation~\eqref{eqn:simple_sys}. The macroscopic variable $\Xb$ satisfies
\begin{equation}\label{eqn:xbar}
d\Xb = \tfrac{1}{n}\bbone^{\tr}dX_t = -\frac{\nor{\bx}^2}{n}\bbone^{\tr}\tanh(X_t)dt +
\frac{\sigmat}{\sqrt{n}}dB_t
\end{equation}
and thus
\begin{equation}\label{eqn:flucts_process}
d\Xt = dX_t - d\Xb\bbone =
 -\Bigl(\tanh(X_t)\nor{\bx}^2 + LX_t - \frac{\nor{\bx}^2}{n}\bbone^{\tr}\tanh(X_t)\bbone \Bigr)dt +
\sigmat d\bB_t - \frac{\sigmat}{\sqrt{n}}dB_t\bbone .
\end{equation}
In terms of the original variable $\bw$, the fluctuations $\widetilde{\bw}_t$ are purely due to the noise, while $\bar{w}_t$ parameterizes the average decision function. As the decision function we consider is linear, the uncertainty in the decisions is directly equivalent to uncertainty in the parameter $w$.
We will study the evolution of both the mean and the fluctuation processes over time, however to assess uncertainty the central quantity of interest will be the size of the ball containing the
fluctuations (the ``microscopic'' variables). We characterize the magnitude of the fluctuations via the squared norm process satisfying
\begin{equation}\label{eqn:flucts}
\frac{d\|\Xt\|^2}{2} = -\Bigl(\nor{\bx}^2\langle \Xt,\tanh(X_t)\rangle + \langle \Xt,LX_t\rangle\Bigr)dt  +
\frac{1}{2}\sigmat^2(n-1)dt + \sigmat\|\Xt\|dB_t
\end{equation}
which follows from~\eqref{eqn:flucts_process} applying Ito's Lemma to the function $h(\Xt) = \frac{1}{2}\scal{\Xt}{\Xt}$ and the fact that $\scal{\Xt}{d\bB_t} = \nor{\Xt}dB_t$ in law.

\subsection{Uncertainty Estimates}
The first --and central-- result says that the ball centered at $\bar{w}$ (the center of mass) containing the fluctuations can be controlled in expectation from above and below by the coupling strength and in most cases the number of circuit copies, via the spectrum of the network Laplacian $L$. We note that lower bounds are typically ignored in the dynamical systems literature, possible because they are less important for stability analyses. We have found, however, that such bounds can be derived in the case of saturated gradient dynamics, and that control from below can yield further insight into the present problem of neural learning.

Let $\lambda_{+}$ be the largest eigenvalue of L, and let $\lambda_{-}$ be the smallest non-zero eigenvalue of $L$.
\begin{theorem}[Fluctuations can be made small]\label{thm:tilde_bounds}
After transients of rate $2\lambda_{-}$
\begin{equation*}
\frac{(n-1)\sigma^2}{2\lambda_{+}}\left(1-\frac{\nor{\bx}^2}{\lambda_{-}}\right) \leq
\bbE\bigl\|\tilde{\bw}(t)\bigr\|^2 \leq \frac{(n-1)\sigma^2}{2\lambda_{-}}
\end{equation*}
where $\widetilde{\bw}=P\bw(t)$.
\end{theorem}
Clearly the lower bound is informative only when $\lambda_{-} > \nor{\bx}^2$. While we do not
explicitly assume any particular bound on the size of the examples $\nor{\bx}$, it is reasonable that
$\lambda_{-} \gg \nor{\bx}^2$ since $\lambda_{-}$ can depend on the number of circuits $n$ and will
always depend on the coupling strength $\kappa$, which can be large. Large coupling strengths can be found in a variety of circumstances, particularly in the case of motor control circuits~\citep{Grandhe99,Kiemel03} for example.

In the next Theorem we give the variance of the fluctuations via a higher moment of $\nor{\widetilde{\bw}}$.
This result makes use of the lower bound in Theorem~\ref{thm:tilde_bounds}, and leads to a result that gives control of the fluctuations in probability rather than in expectation.
\begin{theorem}[Variance of the trajectory distances to the center of mass]\label{thm:fluct-var}
After transients of rate $2\lambda_{-}$
\begin{equation*}
\var\bigl(\nor{\tilde{\bw}(t)}^2\bigr) \leq
\left(\frac{(n-1)\sigma^{2}}{2\lambda_{-}}\right)^2\left(2 + \frac{4}{n-1}\right)
- \left(\frac{(n-1)\sigma^{2}}{2\lambda_{+}}\right)^2\left(1-\frac{\nor{\bx}^2}{\lambda_{-}}\right)^2 .
\end{equation*}
\end{theorem}
Chebyshev's inequality combined with Theorem~\ref{thm:fluct-var} immediately gives the following
Corollary.
\begin{corollary}\label{cor:cheby-fluct}
After transients of rate $2\lambda_{-}$
\begin{equation}\label{eqn:cheby-fluct}
\mathbb{P}\left[\bigl\|\tilde{\bw}(t)\bigr\|^2 - \bbE\bigl\|\tilde{\bw}(t)\bigr\|^2 \geq \varepsilon\right]
\leq \frac{\var\bigl(\nor{\tilde{\bw}(t)}^2\bigr)}{\varepsilon^2}.
\end{equation}
\end{corollary}
Since any connected network graph has non-trivial eigenvalues which depend on the uniform coupling strength $\kappa$, we see that for fixed $n$ as $\kappa\to\infty$, $\var\bigl(\nor{\tilde{\bw}(t)}^2\bigr)\to 0$.
In the case of the all-to-all network topology, for example, the eigenvalues of $L$ depend on both $n$ and $\kappa$ so that $\var\bigl(\nor{\tilde{\bw}(t)}^2\bigr) = \mathcal{O}(\kappa^{-2})$ giving a power law decay of order $\mathcal{O}(\kappa^{-2}\varepsilon^{-2})$ on the right hand side of Equation~\eqref{eqn:cheby-fluct} in Corollary~\ref{cor:cheby-fluct}.

Finally, we turn to estimating in expectation the steady-state average distance between the trajectories of the circuit copies and the noise-free solution. As we have argued in Section~\ref{sec:stability}, the rate of convergence to equilibrium of the trajectories $w_i(t)$ can be arbitrarily small. Although from the Theorems above the fluctuations can be made small, one cannot in general make a similar statement about the center of mass $\bar{w}_t$ process unless assumptions about the initial conditions are made (and by extension, the same holds true for the
trajectories $w_i(t)$). Such an assumption would lead to control over the contribution of the $\tanh$ terms, and establishes a lower bound on the contraction rate. Rather than make a specific assumption however, we state a general result: We again provide a lower bound, this time following from the law
of large numbers governing sums of i.i.d. Gaussian random variables and the lower bound on the fluctuations provided by Theorem~\ref{thm:tilde_bounds}.
\begin{theorem}[Average distance to the noise-free trajectory]\label{thm:noise_free}
Denote by $w^{*}$ the minimizer of the squared-error objective~\eqref{eqn:obective_fn}. After transients of rate $2\lambda_{-}$
\[
\frac{\sigma^2}{n} +
\left[\frac{(n-1)\sigma^2}{2n\lambda_{+}}\left(1-\frac{\nor{\bx}^2}{\lambda_{-}}\right)\right]^{+} \leq
 \bbE\left[\frac{1}{n}\sum_{i=1}^n(w_i(t)-w^{*})^2\right] \leq
\frac{\sigma^2}{2\lambda_{-}} + \bbE\bigl[(\bar{w}_t-w^{*})^2\bigr]
\]
where $[\,\cdot\,]^{+}\equiv\max(0,\cdot)$.
\end{theorem}

Theorem~\ref{thm:noise_free} says that average closeness of the noisy system to that of the
noise-free optimum is controlled by the tradeoff between the noise and the coupling strength, and
the number of circuit copies $n$. The former controls in large part the magnitude of the fluctuations, as discussed above. The latter quantity is the unavoidable linear averaging component, and  can be brought to zero only as fast as the law of large numbers allows, $\mathcal{O}(n^{-1/2})$ at best.
For fixed $n$ as $\lambda_{-}^{(n)}\to \infty$, the upper and lower bounds coincide since
$\bbE\bigl[(\bar{w}(t) - w^*)^2\bigr]\to \sigma^2/n$. As both $n\to\infty$ and $\kappa\to\infty$ Theorem~\ref{thm:noise_free} confirms that $w_i(t)\to w^{*}$ in expectation.
If the fluctuations are not made small however, linear averaging will be wrong, and the error will of course be greater. Just how bad linear averaging is when the fluctuations are allowed to be large is described in large part by the maximum curvature of the noise-free dynamics\footnote{One way to see this is to take the first-order Taylor expansion of the dynamics with integral remainder. The remainder term
can be upper bounded by the spectral radius of the Hessian matrix, which is related to curvature (see e.g. ~\citep{Tabareau10}).}.

Finally, we note that the estimates above depend on the number of samples $m$ only through the norm of the examples $\mathbf{x}$, and it is reasonable to assume that this quantity may be appropriately normalized based on the maximum values conveyed by subsystems or rates of neurons comprising the circuit in the case of population or rate codes, or maximum field strengths in the case of LFPs.  However, the requirement that the organism must collect $m$ observations before learning can proceed is not essential. We may also consider the {\em online} learning setting, where data are observed sequentially and updates to the parameters $(w)_i$ are made separately on the basis of each observation in temporal order. The analysis above studies convergence to and distance from the solution in the steady state, whatever that solution may be, given $m$ pieces of evidence. Thus the online setting can also be considered as long as the time between observations is longer than the transient periods. Indeed, in many scenarios learning and decision making processes in the brain can take place on short time scales relative to the time scale on which experience is accumulated. In this case when another piece of information arrives, the system moves to a region defined (stochastically) around a new steady-state. A complication can arise when the new point arrives during the transient period of the previous learning process -- before the system has had a chance to settle, on average, into the new equilibrium -- however we do not attempt to model this
situation here.

\section{Discussion}\label{sec:discussion}
The estimates given in Section~\ref{sec:theorems} quantify the tradeoff between the degree of synchronization and the noise (error), and the role this tradeoff plays in determining the uncertainty of a decision function learned by way of a stochastic, nonlinear dynamics. Estimates both in expectation and in probability were derived. We showed how and where both the coupling strength and the topology of the network of neural ensembles impact the extent to which the noise, and therefore uncertainty due to error, can be controlled. In particular, for most networks (see Section~\ref{sec:graphs_ex}) the effect of the noise can be reduced by either increasing the coupling strength or the number of redundant systems (or both), leading to a steady-state solution that is going to be closer to the ideal, error-free solution with high probability. From a technical standpoint, this is because fluctuations about the common trajectory are exactly the way in which the noise enters the picture; when the fluctuations are made small, the error is made small. In this way an organism may mitigate error imposed by a noisy, imperfect learning apparatus and solve a learning task to greater accuracy. Furthermore, synchronization and redundancy can both improve the {\em speed} of learning, in the sense that the rate of convergence to the steady state solution also depends on these mechanisms. Each of the bounds presented in Section~\ref{sec:theorems} above hold after transient terms of order $e^{-t\lambda_{-}}$ vanish, where $\lambda_{-}$ is the smallest non-zero eigenvalue of the network Laplacian. For any stable connected network, strong coupling strengths directly improve convergence rates to the steady-state, as seen by the dependence of $\lambda_{-}$ on $\kappa$. In the case of all-to-all (including approximately all-to-all and many random graphs), $\lambda_{-} = \mathcal{O}(n\kappa)$ so that both increased redundancy and sync will improve the speed of learning.

Our overarching goal has been to explore quantitatively the role of redundancy and synchronization in reducing error incurred by a stochastic, non-ideal learning and decision making neural process. We have gone about this by considering a model which emphasizes the underlying computations taking place rather than particular neural representations. Looking at the appropriate scale, we seek to address the precise meaning of ensemble measurements and population codes, as well as the information these codes convey about the underlying dynamics and signals. The results derived above support the notion that synchronization and redundancy play a more functional role in the context of learning processes in the brain, rather than being a mere epiphenomenon.

\subsection{Synchronization and Redundancy in the Brain}
Synchronization has been suggested, over a diverse history of experimental work, as a fundamental mechanism for improvement in precision and reduction of uncertainty in the nervous system (see e.g. ~\citep{Needleman:PD:2001,Enright80}). Redundancy too is an important and commonly occurring mechanism. In retinal ganglion cells~\citep{Croner:PNAS:93,Puchalla05} and heart cells~\citep{Clay:79} the spatial mean across coupled cells cancels out noise. Populations of hair cells in otoliths perform redundant, collaborative computations to achieve robustness~\citep{KandelBook,EliasmithBook}, and it has been suggested that multiple cortical (amygdala-thalamus) loops contribute to fear response/conditioning, and emotion more generally~\citep{LeDoux:AnnRevNeuro:00}. With motor tasks such as reaching or standing, it has been argued that planning and representation occurs at least partially in redundant coordinate systems and involve redundant degrees of freedom~\citep{Scholz99}. Todorov~\citep{Todorov:08} maintains that redundancy and noise combine to give rise to optimal muscle control policies, raising the interesting possibility that in some cases the impact of the noise may need to be adjusted but not necessarily eliminated altogether. On a more localized scale, reach direction has also been found to be conveyed by populations of neurons with overlapping tuning curves~\citep{Georgop:82} where synchrony within such populations plays an important role~\citep{Grammont99}. Multiple sensorimotor transformations involving disparate brain regions may be at play in the parietal cortex, where redundant sensory inputs from multiple modalities must be mapped into motor responses~\citep{Ting97,Pouget:97}. In the ascending auditory pathway, varying degrees of redundancy have been noted, and contribute to the robust representation of frequency and more complex auditory objects~\citep{Chechik06}. Ensemble measurements have also been connected to behavior and have been suggested as inputs to brain-machine interfaces, while in stochastic neural decision making it has been suggested that it is the collective behavior across multiple populations of neurons that is responsible for perception and decision making, rather than activity of a single neuron or population of neurons~\citep{Gigante:PLOS:09}.

In these examples and more generally, we suggest that redundancy plus feedback synchronization is a mechanism which may be used to improve the accuracy, robustness and speed of a learning process involving the relevant respective brain areas. This is separate from, and in contrast with, redundancies which are harnessed to specifically increase storage capacity, as in the case of associate memory models~\citep{HertzBook}. There, robustness to corruption is also achieved (via pattern completion dynamics) but the degree of robustness must be traded off against capacity. The primary function of such populations of neurons is to ostensibly store and retrieve memory patterns rather than to implement adaptive, learning dynamics while eliminating noise.

Another theme emerging from these instances of sync and redundancy, is that key computations may be seen as implemented by distant brain regions coupled together by way of long-distance projections and network ``hubs''. Recent experimental observations in {\em C. elegans} casts this interpretation in a  developmental light~\citep{Varier2011}, and suggests that such interactions occur from an early stage in life and are important for normal development in even simple organisms. Learning processes realized by such computations and interactions are certainly susceptible to noise, and must cope with this noise one way or another. We suggest that synchronization and redundancy are not only present and possible, but provide a ready, natural solution.

The ability to learn and make decisions reliably in the presence of uncertainty is of fundamental importance for survival of any organism. This uncertainty can be seen to arise from three distinct sources, and the approach discussed here treats only the first two: intrinsic neuronal noise, both local and in aggregate, and noise in the form of measurement error, under which we include error due to limitations in precision and nonlinearity in biological systems. A third and equally important source of error is that of uncertainty in the inference process itself~\citep{Shadlen:Nature:07,Shadlen:Science:09}. This uncertainty is specific to and inherent in the decision problem and is characterized by the posterior distribution over decisions given the experiential evidence.  Our work only considers uncertainty beyond that of the inference process, and as such is one part of a larger puzzle. We argue that intrinsic noise is both experimentally and theoretically important -- and involved enough technically -- to be addressed in isolation, while holding all other variables constant. Indeed, intrinsic noise intensities can be large. The role of the network's topology and coupling mechanism also strongly influences the overall picture, often in surprising or subtle ways. But it is also possible that the methods recruited here can be applied towards understanding some aspect of the inference error if different inferences from the same observations can be made by different ``expert'' (circuits) each with their own biases. Then averaging, nonlinearity and the uncertainty could potentially be treated in a similar framework.

\subsection{Extensions and Generalizations}
Asymptotic stability of the stochastic system considered here is guaranteed as long as there is coupling.
In general, if the dynamics of a stochastic system are contracting or can be made contracting with feedback, then combinations (e.g. parallel, serial, hierarchical) of such systems will be contracting~\citep{Pham09,Lohmiller98}. In the present setting, the system governing the fluctuations about the mean trajectory is contracting with a rate dependent on the coupling strength
and the noise variance. Thus combinations of learning systems of the general type considered here can enjoy strong stability guarantees automatically, since the individual systems are contracting.

Finally, we have assumed throughout that the errors affecting the collection of redundant neural circuits or systems are mutually independent. This is not an unreasonable modeling assumption: For large-scale learning processes involving different brain areas, noise imposed by local spike irregularities is largely unrelated to noise present in distant circuits. Within small populations of neurons, it is likely that dependence among intrinsic neuronal noise sources decays rapidly in space so that nearest-neighbors may experience somewhat correlated noise, but beyond this are not significantly impacted by other members of the population. As noises in a biological environment can never be fully dependent (whether due to thermal or chemical-kinetic factors, or otherwise), partial-dependence among noise inputs may be explicitly modeled as, for example, mixing processes if desired~\citep{DoukhanMixing}. Estimates of the form discussed here would then be augmented with mixing terms leading to results which make identical qualitative statements about the role of redundancy and sync. Fluctuations, and the effect of the noise, would still be reducible  but would require larger coupling strengths or more redundancy compared to what would be necessary if the noise sources were independent.

\begin{figure}[t]
\centering
\includegraphics[height=5.5cm]{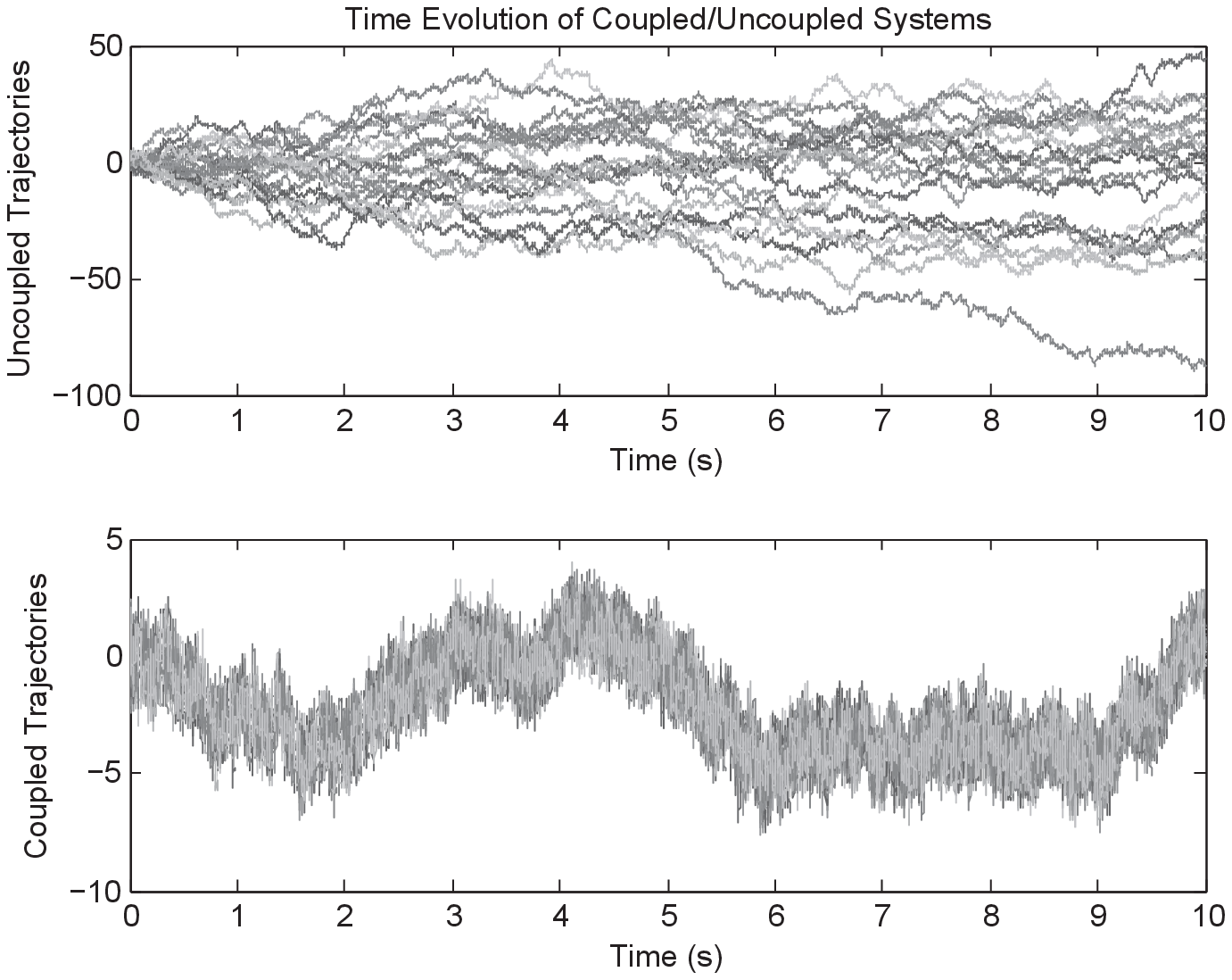}
\hskip 0.1cm
\includegraphics[height=5.5cm]{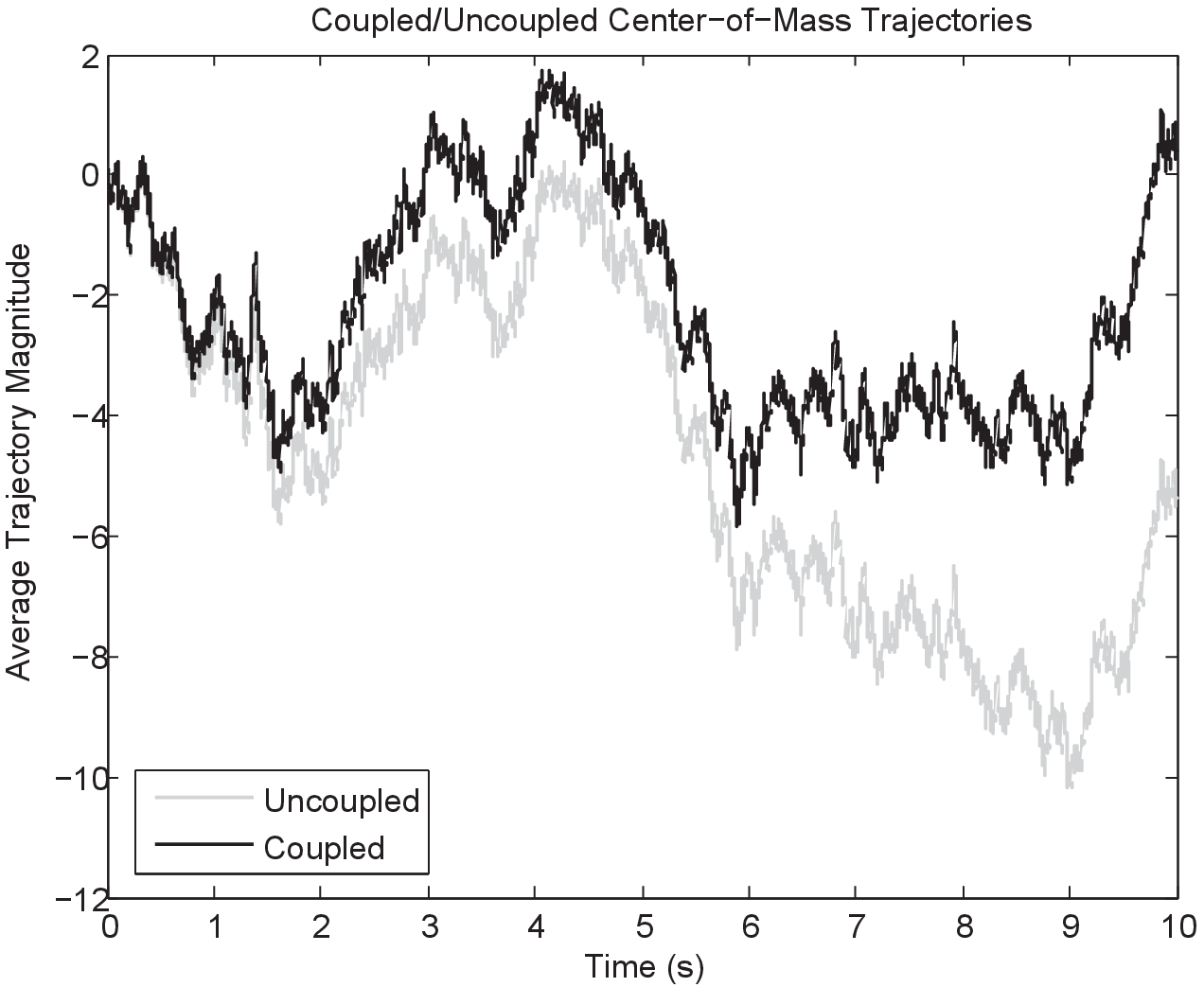}
\caption{{\em (Left) Typical simulated trajectories for coupled and uncoupled networks driven by the same noise. (Right) Population average trajectories for the coupled and uncoupled systems.}}
\label{fig:sims}
\end{figure}

\section{Simulations}\label{sec:expts}
To empirically test the estimates given in Section~\ref{sec:theorems}
we simulated several systems of SDEs given by Equation~\eqref{eqn:simple_sys} using Euler-Maruyama
integration (over time $t\in[0,10s]$, $10^5$ regularly spaced sample points), for different settings of the parameters $n$ (number of circuits or elements), $\kappa$
(coupling strength) and $\sigma$ (noise standard deviation). Initial conditions were randomly drawn from the uniform distribution on $[-5,5]$, and we fixed $\nor{\bx}^2=1$ and the coupling arrangement to all-to-all coupling with fixed strength determined by $\kappa$. For simplicity the simulated systems
had equilibrium point at zero, corresponding to $\by=0$, so that $\scal{x}{y}=0$ and $X^*=w^*$ (the
change of variables is the identity map and we can identify $X_t$ with $w_t$).

For comparison purposes we first show on the left in Figure~\ref{fig:sims} typical simulated trajectories of uncoupled (top) and coupled (bottom) populations when $n=20, \kappa=5, \sigma=10$. Both populations are driven by the same noise and the same set of initial conditions, however each element is driven by noise independent from the others as assumed above. From the units on the vertical axes, one can see that coupling clearly reduces inter-trajectory fluctuations as expected. On the right in Figure~\ref{fig:sims}, we show the coupled/uncoupled populations' respective center of mass trajectories for this particular simulation instance. One can see from this figure that the average of the coupled system tends closer to zero ($X^*$), and is less affected by large noise excursions.

\begin{table}[p]
\centering
\begin{tabularx}{0.9\textwidth}{ c | c | X | X}
Quantity & Lower Bound & Simulated & Upper Bound \\ \hline\hline
$\bbE\nor{\Xt}^2$ & 9.405 & 9.497 (std$=3.1$) & 9.500 \\  
$\var\bigl(\nor{\Xt}^2\bigr)$ & - & 9.450 (std$=14.7$) & 111.046 \\  
$\tfrac{1}{n}\bbE\nor{X_t - X^*\bbone}^2$ & 5.470 & 12.249 (std$=22.2$) & 12.249 (std$=22.2$) \\  
\end{tabularx}
\caption{{\em Estimates vs. simulated quantities: $n=20, \kappa=5, \sigma=10$.}}
\label{tab:sims_20_5_10}
\end{table}

\begin{table}[p]
\centering
\begin{tabularx}{0.9\textwidth}{ c | c | X | X}
Quantity & Lower Bound & Simulated & Upper Bound \\ \hline\hline
$\bbE\nor{\Xt}^2$ & 11.281 & 11.719 (std$=3.8$) & 11.875 \\  
$\var\bigl(\nor{\Xt}^2\bigr)$ & - & 14.261 (std$=23.0$) & 184.45 \\  
$\tfrac{1}{n}\bbE\nor{X_t - X^*\bbone}^2$ & 1.814 & 1.933 (std$=2.5$)  & 1.946 (std$=2.4$)\\  
\end{tabularx}
\caption{{\em Estimates vs. simulated quantities: $n=20, \kappa=1, \sigma=5$.}}
\label{tab:sims_20_1_5}
\end{table}

\begin{table}[p]
\centering
\begin{tabularx}{0.9\textwidth}{ c | c | X | X}
Quantity & Lower Bound & Simulated & Upper Bound \\ \hline\hline
$\bbE\nor{\Xt}^2$ & 45.125 & 47.053 (std$=15.2$) & 47.500 \\  
$\var\bigl(\nor{\Xt}^2\bigr)$ & - & 230.275  (std$=373.4$)& 2951.234 \\  
$\tfrac{1}{n}\bbE\nor{X_t - X^*\bbone}^2$ & 7.256 & 14.761 (std$=24.1$) & 14.784 (std$=24.1$)\\  
\end{tabularx}
\caption{{\em Estimates vs. simulated quantities: $n=20, \kappa=1, \sigma=10$.}}
\label{tab:sims_20_1_10}
\end{table}

\begin{table}[p]
\centering
\begin{tabularx}{0.9\textwidth}{ c | c | X | X}
Quantity & Lower Bound & Simulated & Upper Bound \\ \hline\hline
$\bbE\nor{\Xt}^2$ & 49.005 & 49.556 (std$=7.0$) & 49.500 \\  
$\var\bigl(\nor{\Xt}^2\bigr)$ & - & 49.332 (std$=70.6$) & 2598 \\  
$\tfrac{1}{n}\bbE\nor{X_t - X^*\bbone}^2$ & 1.490 & 1.449 (std$=1.6$) & 1.449 (std$=1.6$) \\  
\end{tabularx}
\caption{{\em Estimates vs. simulated quantities: $n=100, \kappa=1, \sigma=10$.}}
\label{tab:sims_100_1_10}
\end{table}

\begin{table}[p]
\centering
\begin{tabularx}{0.9\textwidth}{ c | c | X | X}
Quantity & Lower Bound & Simulated & Upper Bound \\ \hline\hline
$\bbE\nor{\Xt}^2$ & 9.880 & 10.137 (std$=1.5$) & 9.900 \\  
$\var\bigl(\nor{\Xt}^2\bigr)$ & - & 2.151 (std$=3.2$)& 102.362 \\  
$\tfrac{1}{n}\bbE\nor{X_t - X^*\bbone}^2$ & 1.099 & 1.496 (std$=1.5$)  & 1.496 (std$=1.5$) \\  
\end{tabularx}
\caption{{\em Estimates vs. simulated quantities: $n=100, \kappa=5, \sigma=10$.}}
\label{tab:sims_100_5_10}
\end{table}

To empirically test tightness of the estimates given in Section~\ref{sec:theorems}, we
repeated simulations of each respective system $5000$ times, and averaged the relevant outcomes
to approximate the expectations appearing in the bounds. Transient periods were excluded in all
cases. In Tables~\ref{tab:sims_20_5_10} through~\ref{tab:sims_100_5_10} we show the values predicted by the bounds and the corresponding simulated quantities, for each respective triple of system parameter settings. Sample standard deviations of the simulated averages (expectations) are given in parentheses. In Figure~\ref{fig:sim_flucts} we show theoretical versus simulated expected magnitudes of the fluctuations $\bbE\nor{\Xt}^2$ when $n=200$ and $\sigma=10$ over a range of coupling strengths. The solid dark trace is the upper bound of Theorem~\ref{thm:tilde_bounds}, while the open circles are the average simulated quantities (again $5000$ separate simulations were run for each $\kappa$). Error bars are also given for the simulated expectations. Note that the magnitude scale ($y$-axis) is logarithmic, so the error bars are also plotted on a $\log$ scale. We omitted the lower theoretical bound from the plot because it is too close to the upper bound to visualize well relative to the scale of the bounds.

\begin{figure}[t]
\centering
\includegraphics[height=9cm]{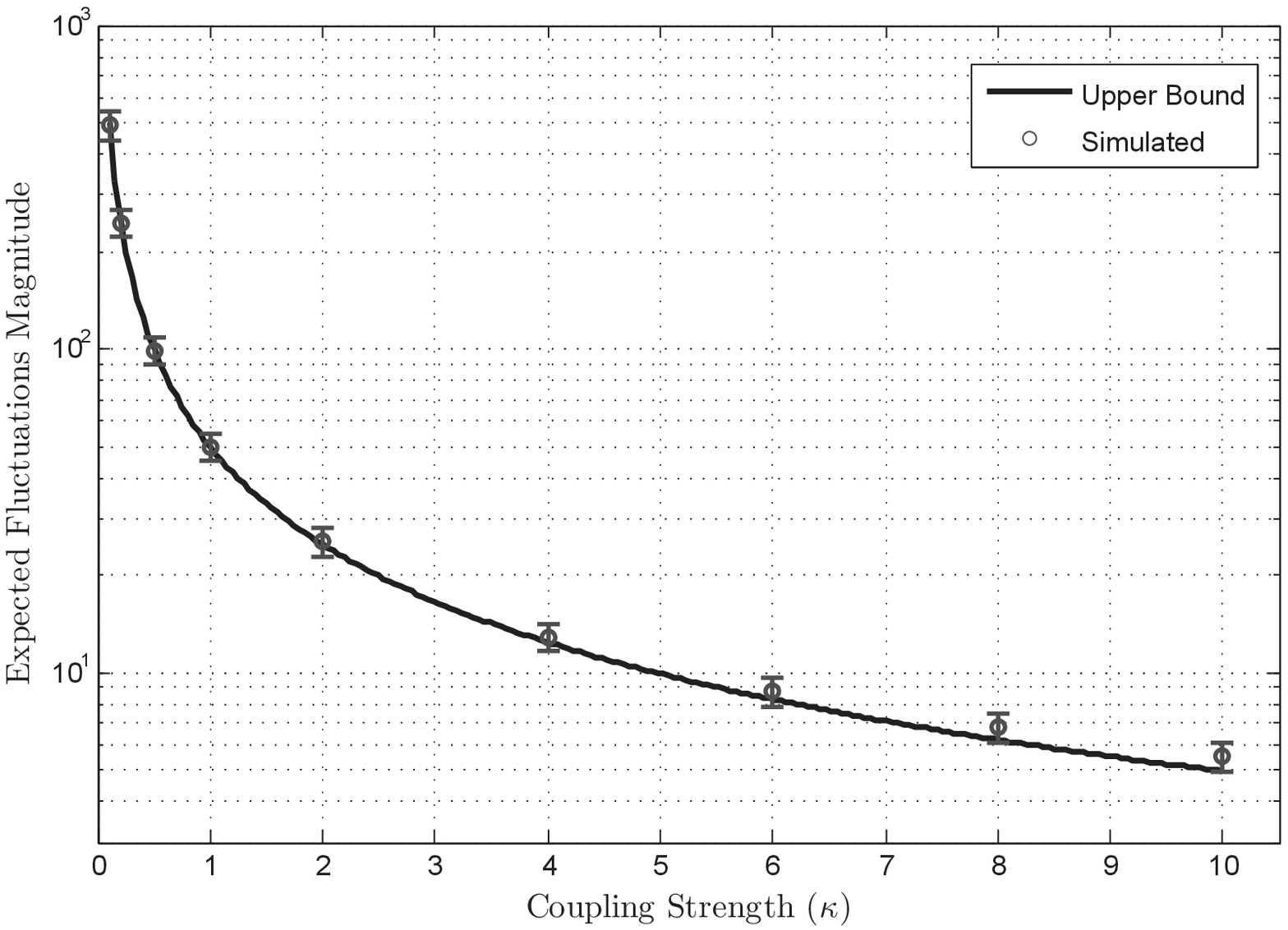}
\caption{{\em Simulated vs. theoretical upper bound estimates of the fluctuations' expected magnitude over a range of coupling strengths $\kappa$. Here $n=200$ circuits and $\sigma=10$.}}
\label{fig:sim_flucts}
\end{figure}

 Generally, the estimates relating to the magnitude of the fluctuations are seen to be tight, and the variance estimate is within an order of magnitude. For the experiments with large noise amplitudes, the empirical estimates can appear to slightly violate the bounds where the bounds are tight since the variance across simulations is large. The lower bound estimating the distance of the center of mass to the noise-free solution is also seen to be reasonably good. For comparison, we give the upper estimate where the empirical distance is substituted in place of the expectation in order to show closeness to the lower bound. Theorem~\ref{thm:noise_free} predicts that the upper and
lower estimates will eventually coincide if $\kappa$ and/or $n$ are chosen large enough.

\section{Proofs}\label{sec:proofs}
In this section we provide proofs of the results discussed in Section~\ref{sec:theorems}.

We first introduce a key Lemma to be used in the development immediately below.
\begin{lemma}\label{lem:tanh_scal}
Let $P = I - (1/n)\bbone\bbone^{\tr}$, the canonical projection onto the zero mean subspace of $\bbR^n$.
Then for all $x\in \bbR^n$
\[
0 \leq \scal{Px}{\tanh(x)} \leq \nor{Px}^2,
\]
where the hyperbolic tangent applies elementwise.
\end{lemma}
\begin{proof}
 Given $x\in\bbR^n$, define the index sets $I=\{1,\ldots,n\}$,
$I_{+}=\{i\in I~|~(Px)_i \geq 0\}$, and $I_{-}= I\setminus I_{+}$. Since $Px$ is zero mean,
$\sum_{i\in I_{+}}(Px)_i = \sum_{i\in I_{-}}|(Px)_i|$.
We will express the hyperbolic tangent as
$\tanh(z) = 2s(2z) - 1$, where
$s(z) = (1 + e^{-z})^{-1}$ is the logistic sigmoid function.
If we let $\mu = \tfrac{1}{n}\bbone^{\tr}\!x$ be the center of mass of $x$,
$(Px)_i = x_i - \mu \geq 0$ implies $s(x_i) \geq s(\mu)$ by monotonicity of s. Likewise,
$(Px)_i < 0$ implies $s(x_i) < s(\mu)$. Finally, note that since $P^2 = P$ and $\bbone\in\ker P$,
$
\scal{Px}{\tanh(x)} = \scal{Px}{P\bigl(2s(2x) - \bbone\bigl)} =
 2\scal{Px}{s(2x)}.
$
Using these facts, we prove the lower bound first:
\begin{align*}
\scal{Px}{\tanh(x)}  & = 2\sum_{i\in I_{+}}(Px)_i s(2x_i) -
2\sum_{i\in I_{-}}|(Px)_i| s(2x_i) \\
&\geq 2s(2\mu)\sum_{i\in I_{+}}(Px)_i - 2s(2\mu)\sum_{i\in I_{-}}|(Px)_i| \\
&= 2s(2\mu)\cdot 0 = 0.
\end{align*}
Turning to the upper bound, we prove the equivalent statement
$\scal{Px}{s(2x) - x}\leq 0$. First, if $\mu=0$, then $Px=x$ so
$\scal{Px}{\tanh(x)}=\scal{x}{\tanh(x)}\leq \nor{x}\nor{\tanh(x)}\leq\nor{x}^2 = \nor{Px}^2$, since
$\nor{\tanh(x)}\leq\nor{x}$ by virtue of the fact that $|\tanh(z)|=\tanh(|z|)\leq |z|$ for any $z\in\bbR$.
Now suppose that $\mu > 0$. If $z \geq \mu > 0$, we can upper bound $s(2z)$ by the line tangent to the point
$(\mu,s(2\mu))$: $s(2z) \leq mz + b$ with $m<\tfrac{1}{2}$ and $b>\frac{1}{2}$. If $z < \mu$, we
can take the lower bound $s(2z) > \tfrac{1}{2}z + \tfrac{1}{2}\mu - s(2\mu)$. Using these estimates, we have
that
\begin{align*}
\scal{Px}{s(2x) - x} &= \sum_{i\in I_{+}}(Px)_i \bigl(s(2x_i) - x_i\bigr) + \sum_{i\in I_{-}}|(Px)_i|\bigl(x_i - s(2x_i)\bigr) \\
&\leq \sum_{i\in I_{+}}(Px)_i \bigl(b - (1-m)x_i) +
\sum_{i\in I_{-}}|(Px)_i|\bigl(\tfrac{1}{2}x_i +\tfrac{1}{2}\mu -s(2\mu)\bigr) \\
&\leq \sum_{i\in I_{+}}(Px)_i \bigl(b - (1-m)\mu) +
\sum_{i\in I_{-}}|(Px)_i|\bigl(\tfrac{1}{2}\mu +\tfrac{1}{2}\mu -s(2\mu)\bigr) \\
&= \Bigl(\sum_{i\in I_{+}}(Px)_i\Bigr)\bigl(b + m\mu - s(2\mu)\bigr) = 0.
\end{align*}
The second inequality follows from the fact that $(1-m)>0$, $x_i \geq \mu$ for $i\in I_{+}$ and $x_i < \mu$
for $i\in I_{-}$. Since $\sum_{i\in I_{+}}(Px)_i = \sum_{i\in I_{-}}|(Px)_i|$, and recalling that by definition $b$
satisfies $m\mu + b = s(2\mu)$, the final equalities follow. If $\mu < 0$, then the proof is similar, taking
the line tangent to the point $(\mu,s(2\mu))$ as a lower bound for $s(2z)$ and the line
$\tfrac{1}{2}(z - \mu) +s(2\mu)$ as an upper bound.
\end{proof}

\subsection{Fluctuations Estimates: Proof of Theorem~\ref{thm:tilde_bounds}}
We begin by adding $\lambda\nor{X_t}^2dt$, with $\lambda\in(0,\infty)$, to both sides of
Equation~\eqref{eqn:flucts} to obtain
\begin{align*}
\tfrac{1}{2}d\nor{\Xt}^2 + \lambda\nor{\Xt}^2dt &= -\nor{\bx}^2\scal{\tanh X_t}{\Xt}dt +
(\lambda\nor{\Xt}^2 - \scal{LX_t}{\Xt})dt \\
& \quad + \frac{1}{2}(n-1)\sigmat^2dt + \sigmat\nor{\Xt}dB_t 
= e^{-2\lambda t}d(\tfrac{1}{2}\nor{\Xt}^2e^{2\lambda t}) ,
\end{align*}
where the second equality follows noticing that the right hand side is the total Ito derivative
of the left hand side of the first equality. Now multiply both sides by $e^{2\lambda t}$,
switch to integral form, and multiply both sides by $e^{-2\lambda t}$  to arrive at
\begin{equation}\label{eqn:int_form}
\begin{split}
\tfrac{1}{2}\nor{\Xt}^2 &= e^{-2\lambda t}\nor{\widetilde{X}_0}^2 +
\int_{0}^{t} e^{2\lambda(s-t)}\Bigl(\frac{1}{2}(n-1)\sigmat^2 -
\nor{\bx}^2\scal{\tanh X_s}{\widetilde{X}_s}\Bigr)ds \\
 &+ \int_{0}^{t} e^{2\lambda(s-t)}(\lambda\nor{\widetilde{X}_s}^2 -\scal{LX_s}{\widetilde{X}_s})ds +
\sigmat\int_{0}^{t} e^{2\lambda(s-t)}\nor{\widetilde{X}_s}dB_s .
\end{split}
\end{equation}

\noindent{\bf Upper Bound:}
Next, note that $\scal{LX_t}{\Xt} = \scal{L\Xt}{\Xt}$ since $L(\Xb\bbone)=0$, and that
$\Xt$ is by definition orthogonal to any constant vector. For all $t$ we also have that
\begin{equation}\label{eqn:upper_ests}
\begin{aligned}
\lambda_{-}\nor{\Xt}^2 -\scal{L\Xt}{\Xt} & \leq 0 \\
-\scal{\tanh X_t}{\Xt} &\leq 0
\end{aligned}
\end{equation}
almost surely. The first inequality follows from the fact that for all $x\in\im P$,
\[
\lambda_{-}\nor{\Xt}^2\leq\scal{L\Xt}{\Xt}\leq\lambda_{+}\nor{\Xt}^2,
\]
if $\lambda_{-}$ is the Fiedler eigenvalue of $L$ and $\lambda_{+}$ is the largest eigenvalue of $L$.
The second inequality is given by Lemma~\ref{lem:tanh_scal}.
Setting $\lambda\equiv\lambda_{-}$ and applying the inequalities~\eqref{eqn:upper_ests} to
Equation~\eqref{eqn:int_form} gives the estimate
\begin{align}\label{eqn:flucts_beforeE}
\tfrac{1}{2}\nor{\Xt}^2 &\leq e^{-2\lambda_{-}t}\nor{\widetilde{X}_0}^2 +
\frac{(n-1)\sigmat^2}{2}\int_{0}^{t} e^{2\lambda_{-}(s-t)}ds +
\sigmat\int_{0}^{t} e^{2\lambda_{-}(s-t)}\nor{\widetilde{X}_s}dB_s \nonumber \\
&= e^{-2\lambda_{-}t}\nor{\widetilde{X}_0}^2 + \frac{(n-1)\sigmat^2}{4\lambda_{-}}(1-e^{-2\lambda_{-}t})  +
\sigmat\int_{0}^{t} e^{2\lambda_{-}(s-t)}\nor{\widetilde{X}_s}dB_s
\end{align}
almost surely.
Taking expectations and noting that
$\bbE\left[\int_{0}^{t} e^{2\lambda_{-}(s-t)}\nor{\widetilde{X}_s}dB_s\right] = 0,
$
 we have that
\begin{equation}\label{eqn:bnd_upper}
\bbE\nor{\Xt}^2 \leq \frac{(n-1)\sigmat^2}{2\lambda_{-}}
\end{equation}
after transients of rate $2\lambda_{-}$.

\noindent{\bf Lower Bound:}
We show that $\bbE \nor{\Xt}^2$ has a lower bound that can also be expressed in terms of the coupling strength and the noise level.
The derivation is similar to that of the upper bound, and we begin with Equation~\eqref{eqn:int_form}.
We set
$\lambda\equiv\lambda_{+}$ and apply the estimates $\lambda_{+}\nor{\widetilde{X}_s}^2 -\scal{L\widetilde{X}_s}{\widetilde{X}_s} \geq 0$ and $\scal{\tanh X_s}{\widetilde{X}_s} \leq \nor{\widetilde{X}_s}^2$
for all $s$ a.s., yielding
\begin{equation*}
\tfrac{1}{2}\nor{\Xt}^2 \geq e^{-2\lambda_{+}t}\nor{\widetilde{X}_0}^2 +
\int_{0}^{t} e^{2\lambda_{+}(s-t)}\Bigl(\frac{1}{2}(n-1)\sigmat^2 - \nor{\bx}^2\nor{\widetilde{X}_s}^2\Bigr)ds +
\sigmat\int_{0}^{t} e^{2\lambda_{+}(s-t)}\nor{\widetilde{X}_s}dB_s .
\end{equation*}
Taking expectations and integrating the Ito term, we have
\begin{equation*}
\tfrac{1}{2}\bbE\nor{\Xt}^2 \geq
e^{-2\lambda_{+}t}\bbE\nor{\widetilde{X}_0}^2 +
\frac{(n-1)\sigmat^2}{4\lambda_{+}}(1-e^{-2\lambda_{+}t}) -
\nor{\bx}^2\int_{0}^{t} e^{2\lambda_{+}(s-t)}\bbE\nor{\widetilde{X}_s}^2 ds .
\end{equation*}
After transients of rate $2\lambda_{-}$, we can apply~\eqref{eqn:bnd_upper} to estimate the remaining integral and
lower bound the above equation by
\[
 e^{-2\lambda_{+}t}\bbE\nor{\widetilde{X}_0}^2 +
\frac{(n-1)\sigmat^2}{4\lambda_{+}}(1-e^{-2\lambda_{+}t}) -
\nor{\bx}^2\frac{(n-1)\sigmat^2}{4\lambda_{-}\lambda_{+}}(1-e^{-2\lambda_{+}t}) .
\]
Since $\lambda_{-}\leq\lambda_{+}$, transients of rate $2\lambda_{+}$ have already transpired if we
suppose that we have waited for transients of rate $2\lambda_{-}$. Therefore, we can say that
after transients of rate $2\lambda_{-}$,
\begin{equation}\label{eqn:bnd_lower}
\bbE\nor{\Xt}^2 \geq \frac{(n-1)\sigmat^2}{2\lambda_{+}}\left(1-\frac{\nor{\bx}^2}{\lambda_{-}}\right).
\end{equation}

\subsubsection{Inverting the change of variables}\label{sec:reverse_ch}
Finally, we can obtain corresponding upper and lower bounds for the original system~\eqref{eqn:orig_sys} noting that since $\Xt = P\bigl(\bw(t)\nor{\bx}^2 - \scal{\bx}{\by}\bbone\bigr) = \nor{\bx}^2P\bw(t)$, we have $\bbE\nor{\tilde{\bw}}^2 = \bbE\nor{\Xt}^2/\nor{\bx}^4$, where we have used the notation $\tilde{\bw}$ for $P\bw$. The $\nor{\bx}^4$ in the denominator then cancels with the same quantity occurring in $\sigmat^2$ in Equations~\eqref{eqn:bnd_upper} and~\eqref{eqn:bnd_lower}, giving the final form shown in Theorem~\ref{thm:tilde_bounds}.

\subsection{Fluctuations Estimates: Proof of Theorem~\ref{thm:fluct-var}}
We first derive the fourth moment of the norm of the fluctuations.
Starting from Equation~\eqref{eqn:flucts_beforeE}, allow transients of rate $2\lambda_{-}$ to pass so
that we are left with the integral inequality
\[
\tfrac{1}{2}\nor{\Xt}^2 \leq  \frac{(n-1)\sigmat^2}{4\lambda_{-}} +
\sigmat\int_{0}^{t} e^{2\lambda_{-}(s-t)}\nor{\widetilde{X}_s}dB_s .
\]
Squaring both sides, we can apply the identity $(a+b)^2\leq 2a^2 + 2b^2$
to obtain
\[
\nor{\Xt}^4 \leq \left(\frac{(n-1)\sigmat^2}{\sqrt{2}\lambda_{-}}\right)^2 +
8\sigmat^2\left(\int_{0}^{t}e^{2\lambda_{-}(s-t)}\nor{\widetilde{X}_s}dB_s\right)^2 .
\]
Taking expectations and invoking Ito's Isometry for the second term leads to
\begin{align*}
\bbE\nor{\Xt}^4 &\leq \left(\frac{(n-1)\sigmat^2}{\sqrt{2}\lambda_{-}}\right)^2 +
8\sigmat^2\int_{0}^{t}e^{4\lambda_{-}(s-t)}\bbE\nor{\widetilde{X}_s}^2ds \\
 &\leq \left(\frac{(n-1)\sigmat^2}{\sqrt{2}\lambda_{-}}\right)^2 +
\frac{8\sigmat^2}{4\lambda_{-}}\left(\frac{(n-1)\sigmat^2}{2\lambda_{-}}\right)
= \left(\frac{(n-1)\sigmat^2}{2\lambda_{-}}\right)^2\left(2 + \frac{4}{n-1}\right)
\end{align*}
where the estimate~\eqref{eqn:bnd_upper} has been substituted in for $\bbE\nor{\widetilde{X}_s}^2$.
An upper bound on the variance is then obtained from the identity $\var(Z^2)=\bbE[Z^4] - (\bbE Z^2)^2$ and the lower estimate given in Equation~\eqref{eqn:bnd_lower}. Reversing the change of variables as in
Section~\ref{sec:reverse_ch} yields the final result.

\subsection{Distance to the Noise-Free Trajectory: Proof of Theorem~\ref{thm:noise_free}}
Theorem~\ref{thm:tilde_bounds} can be applied towards providing a lower bound for the average distance between the noisy trajectories of the neural circuit and the noise-free solution to the learning problem. First observe that from the orthogonal decomposition $X_t = PX_t + QX_t$ and the change of variables
mapping~\eqref{eqn:orig_sys} to~\eqref{eqn:simple_sys},
\begin{equation}\label{eqn:norm_decomp}
\nor{X_t}^2 = \nor{\Xb\bbone}^2 + \nor{\Xt}^2 = \nor{\bx}^4\nor{\bw - w^{*}\bbone}^2.
\end{equation}
Furthermore, we have that
\[
\Xb = n^{-1}\sum_iX_i(t) = n^{-1}\sum_i(w_i\nor{\bx}^2 - \scal{\bx}{\by}),
\]
so evidently $\nor{\bx}^{-4}\bbE\Xb^2 = \bbE[(\bar{w}_t - w^*)^2]$.
Next, note that if the fluctuations are small, the trajectories $(w_i(t))_{i=1}^n$ are close
to one another and the average trajectory $\bar{w}_t = n^{-1}\bw(t)^{\tr}\bbone$ evolves essentially
as $\bar{w}_t\sim w^* + \frac{\sigma}{\sqrt{n}}W_t$, where $W_t$ is interpreted as a white noise
process. In this case we then have that $\bbE[(\bar{w}_t-w^*)^2] = \frac{\sigma^2}{n}$, and we see that
$\bbE[(\bar{w}_t-w^*)^2] \geq \frac{\sigma^2}{n}$ when the fluctuations are not necessarily small.
So we have that $\nor{\bx}^{-4}\bbE\Xb^2 \geq \frac{\sigma^2}{n}$.
Combining the above with Theorem~\ref{thm:tilde_bounds},
\[
\frac{\sigma^2}{n} +
\left[\frac{(n-1)\sigma^2}{2n\lambda_{+}}\left(1-\frac{\nor{\bx}^2}{\lambda_{-}}\right)\right]^{+} \leq
\frac{\bbE\Xb^2}{\nor{\bx}^4} + \frac{\bbE\nor{\Xt}^2}{n\nor{\bx}^4} \leq
\frac{\sigma^2}{2\lambda_{-}} + \bbE\bigl[(\bar{w}_t-w^{*})^2\bigr]
\]
with the notation $[\,\cdot\,]^{+}\equiv\max(0,\cdot)$. Equation~\eqref{eqn:norm_decomp} then shows that
the middle quantity above is equal to $\bbE\left[\frac{1}{n}\sum_{i=1}^n(w_i(t)-w^{*})^2\right]$.

\subsection*{Acknowledgments}
The authors acknowledge helpful suggestions from and discussions with Jonathan Mattingly.
JB gratefully acknowledges support under NSF contract IIS-08-03293, ONR contract N000140710625 and
Alfred P. Sloan Foundation grant no. BR-4834 to M. Maggioni.


\bibliographystyle{apalike}

\end{document}